\documentclass[twocolumn,conference,10pt]{IEEEtran}

\usepackage{amsmath, amssymb,graphicx,cite}
\usepackage{srcltx}
\usepackage{epsfig,amsfonts}
\usepackage{graphicx,cite,amssymb,amsmath}
\usepackage{xcolor}
\usepackage[nolist]{acronym}
\usepackage{psfrag}
\usepackage{perso}

\usepackage{subfigure}
\usepackage{flushend}
\ifCLASSOPTIONcompsoc

\else

\fi
\usepackage{framed}
\usepackage{amsthm}
\usepackage{relsize}
\usepackage{mathtools}
\mathtoolsset{showonlyrefs}

\IEEEoverridecommandlockouts

\begin{document}
\newcommand{\Amatrix}{\mathbf{A}}
\renewcommand{\Bmatrix}{\mathbf{B}}
\newcommand{\Cmatrix}{\mathbf{C}}
\newcommand{\Dmatrix}{\mathbf{D}}

\renewcommand{\c}{\mathsf{c}}
\renewcommand{\b}{\mathsf{b}}
\renewcommand{\r}{\mathsf{r}}
\newcommand{\n}{\mathsf{n}}

\newcommand{\res}{\mathtt{r}}
\newcommand{\inact}{\mathtt{i}}
\newcommand{\act}{\mathtt{a}}
\newcommand{\actset}{\mathscr{A}}
\newcommand{\rippleset}{\mathscr{R}}
\newcommand{\cloudset}{\mathscr{C}}
\newcommand{\image}{\mathscr{I}}
\newcommand{\comp}{\bar{\mathscr{I}}}

\newcommand{\kres}{K^\mathtt{r}}
\newcommand{\kinact}{K^\mathtt{i}}
\newcommand{\kact}{K^\mathtt{a}}

\newcommand{\x}{\mathtt{x}}

\newcommand{\osymb}{c}

\newcommand{\ripple}[1]{ \msr{R}_{#1}}
\newcommand{\cloud}[1]{\msr{C}_{#1}}

\newcommand{\Ripp}{\mathsf{R}}
\newcommand{\Ru}{\Ripp_u}
\newcommand{\ru}{\mathsf{r}_u}
\newcommand{\Cloud}{\mathsf{C}}
\newcommand{\Cu}{\mathsf{C}_u}

\newcommand{\cu}{\mathsf{c}_u}
\newcommand{\bu}{\mathsf{b}_u}
\newcommand{\Nu}{\mathsf{N}_u}
\renewcommand{\nu}{\mathsf{n}_u}
\renewcommand{\ni}{\mathsf{Q}}

\renewcommand{\S}[1]{\mathsf{S}_{#1}}
\newcommand{\C}[1]{\mathsf{C}_{#1}}
\newcommand{\R}[1]{\Ripp_{#1}}
\newcommand{\N}[1]{\mathsf{N}_{#1}}

\newcommand{\Brv}{\mathsf{B}}
\newcommand{\Erv}{\mathsf{A}}
\newcommand{\erv}{\mathsf{a}}

\newcommand{\Y}{\mathsf{N}}
\newcommand{\y}{\mathsf{n}}

\newcommand{\dmax}{ d_{\max}}

\renewcommand{\deg}{\mathrm{deg}}

\newcommand{\Ninact }{\mathbb{E}[\Y]}

\newcommand{\nb}[1]{\mathcal{N}\left(#1\right)}
\newcommand{\sz}[1]{\left|#1\right|}

\newcommand{\myop}[1]{%
  \mathchoice{\raisebox{8pt}{$\displaystyle #1$}}
             {\raisebox{8pt}{$#1$}}
             {\raisebox{4pt}{$\scriptstyle #1$}}
             {\raisebox{1.6pt}{$\scriptscriptstyle #1$}}}

\newtheorem{mydef}{Definition}
\newtheorem{example}{Example}
\newtheorem{prop}{Proposition}
\newtheorem{algo}{Algorithm}

\begin{acronym}
\acro{LT}{Luby transform}
\acro{BP}{belief propagation}
\acro{LRFC}{linear random fountain code}
\acro{ML}{maximum likelihood}
\acro{GE}{Gaussian elimination}
\acro{RSD}{robust soliton distribution}
\acro{SA}{simulated annealing}
\acro{ID}{Inactivation Decoding}
\acro{RI}{random inactivation}
\acro{MWI}{maximum weight inactivation}
\acro{PMF}{probability mass function}
\end{acronym}

\newcommand{\figw}{0.99\columnwidth}
\newcommand{\fran}{\textcolor{red}}
\newcommand{\gian}{\textcolor{green!50!black}}
\newcommand{\imp}{\textcolor{orange}}
\newenvironment{francisco}{\par\color{red}}{\par}
\newcommand{\Omegambms}{\Omega^{\star}}


\title{ ~ \\ \tiny{ ~ \\} \Huge Inactivation Decoding Analysis for LT Codes}

\author{\IEEEauthorblockN{Francisco L\'azaro, Gianluigi Liva}
\IEEEauthorblockA{Institute of Communications and Navigation\\
DLR (German Aerospace Center), Wessling, Germany\\
Email: \{Francisco.LazaroBlasco,Gianluigi.Liva\}@dlr.de}
\and
\IEEEauthorblockN{Gerhard Bauch}
\IEEEauthorblockA{Institute for Telecommunication\\
Hamburg University of Technology, Hamburg, Germany\\
Email: Bauch@tuhh.de}
\thanks{This work has been accepted for publication at the 53rd Annual Allerton Conference on Communication, Control, and Computing, 2015.}
\thanks{\copyright 2015 IEEE. Personal use of this material is permitted. Permission
from IEEE must be obtained for all other uses, in any current or future media, including
reprinting /republishing this material for advertising or promotional purposes, creating new
collective works, for resale or redistribution to servers or lists, or reuse of any copyrighted
component of this work in other works}
}

 \maketitle


\thispagestyle{empty}


\begin{abstract}
We provide two analytical tools to model the inactivation decoding process of LT codes. First, a model is presented which derives the  expected number of inactivations occurring in the decoding process of an LT code. This analysis is then extended allowing the derivation of the distribution of the number of inactivations. The accuracy of the method is verified by  Monte Carlo simulations. The proposed analysis opens the door to the design of LT codes optimized for inactivation decoding.
\end{abstract}

\thispagestyle{empty}


\setcounter{page}{1}



{\pagestyle{plain} \pagenumbering{arabic}}


\section{Introduction}\label{sec:Intro}

Fountain codes \cite{byers02:fountain} are a class of erasure correcting codes which can potentially generate an unbounded number of encoded symbols. This feature makes them very useful when the erasure probability of the communication channel is not known at the transmitter. Fountain codes are also a very efficient solution for reliable multicast/broadcast transmissions.
The first class of practical fountain codes, \ac{LT} codes,  was introduced in \cite{luby02:LT} together with an efficient \ac{BP} erasure decoding algorithm exploiting a bipartite graph representation of the codes.  \ac{BP} decoding of \ac{LT} codes performs remarkably well for long source blocks but its performance degrades remarkably when applied to moderate and short lengths \cite{Maatouk:2012}. Raptor codes were introduced in \cite{shokrollahi06:raptor} as a modification of \ac{LT} codes. They consist of a serial concatenation of an \ac{LT} code with an outer (fixed-rate) code  that is normally chosen to be a high rate erasure correcting code.

LT and Raptor codes are often designed and analyzed assuming \ac{BP} decoding and very large source blocks. However, frequently in practice moderate source block sizes are used. For example, the Raptor codes standardized in \cite{MBMS12:raptor} and \cite{luby2007rfc} assume a  source block length which ranges from $1024$ to $8192$ source symbols. The performance of LT codes under \ac{BP} decoding in the finite length regime was analyzed in \cite{Karp2004,shokrollahi2009theoryraptor,Maatouk:2012}. For short to moderate-length source blocks an efficient \ac{ML} decoding algorithm exists which has a manageable complexity and it is actually widely used in practice \cite{shokrollahi2005systems}, \cite{MBMS12:raptor}. This algorithm is commonly referred to as \emph{inactivation decoding}. The erasure correcting performance of \ac{LT} codes under \ac{ML} decoding was studied in \cite{schotsch:2013}. In \cite{Lazaro:ITW104}  a simple model of inactivation decoding for LT codes was presented which provides an approximation of the number of inactivations needed for decoding. The method in \cite{Lazaro:ITW104} is not always accurate.   In this paper we present an accurate finite length analysis of LT codes under inactivation decoding following the approach used in \cite{Karp2004} for the analysis of BP decoding. The analysis we present is not only able to determine the average number of inactivations, but also to provide the distribution of the number of inactivations. This is important for short LT codes, since in this regime substantial deviations from the average number of inactivations can be expected.\footnote{{In \cite{Chingbats} a finite length analysis of batched sparse codes was introduced, that can also be applied to \ac{LT} codes. This analysis provides the expected number of inactivations, like the analysis in Section~\ref{sec:model_first} in this paper.}}

The paper is organized as follows. In Section  \ref{sec:inact} we introduce inactivation decoding of LT codes. Section~\ref{sec:model_first} introduces a first order analysis of  inactivation decoding of \ac{LT} codes which is able to provide the expected number of inactivations needed for decoding. In Section~\ref{sec:model_second} we outline an analysis that provides the distribution of the number of inactivations.  Finally we present the conclusions to our work in Section~\ref{sec:Conclusions}.

\section{Inactivation Decoding of LT codes}\label{sec:inact}

 A binary \ac{LT} code is considered with $k$ input symbols $\mathbf{v}=(v_1,~v_2,~\ldots, v_k)$.  The output degree distribution is denoted by $\Omega= \{\Omega_1,~\Omega_2,~\Omega_3,~\hdots~\Omega_{d_{\mathrm{max}}}\}$ , where $d_{\mathrm{max}} \leq k$ is the maximum output degree. We assume the receiver has collected $m=k+\delta$ output symbols, $\mathbf{\osymb}=(\osymb_1, \osymb_2, \ldots, \osymb_m)$. The parameter $\delta$ is usually referred to as absolute receiver overhead. We denote by $\epsilon =  m/k -1$ the relative receiver overhead. Decoding consists of solving the linear system of equations
\begin{equation}
\mathbf{\osymb} =\mathbf{v} \mathbf{G}^T
\end{equation}
where $\mathbf{G}$ is the $m \times k$ binary matrix which defines the relation between the input and the output symbols.
For \ac{LT} codes the matrix $\mathbf{G}$ is usually sparse. This sparsity can be exploited to perform \ac{ML} decoding in an efficient way \cite{studio3:RichardsonEncoding,studio3:fekri04it,miller04:bec,shokrollahi2005systems} through a decoding algorithm that is commonly referred to as \emph{inactivation decoding} and that consists of the following steps:
\begin{figure*}[t]
\centering
        \subfigure[Structure of $\mathbf{G}$ after triangularization.]{
               \includegraphics[width=0.66\columnwidth,draft=false]{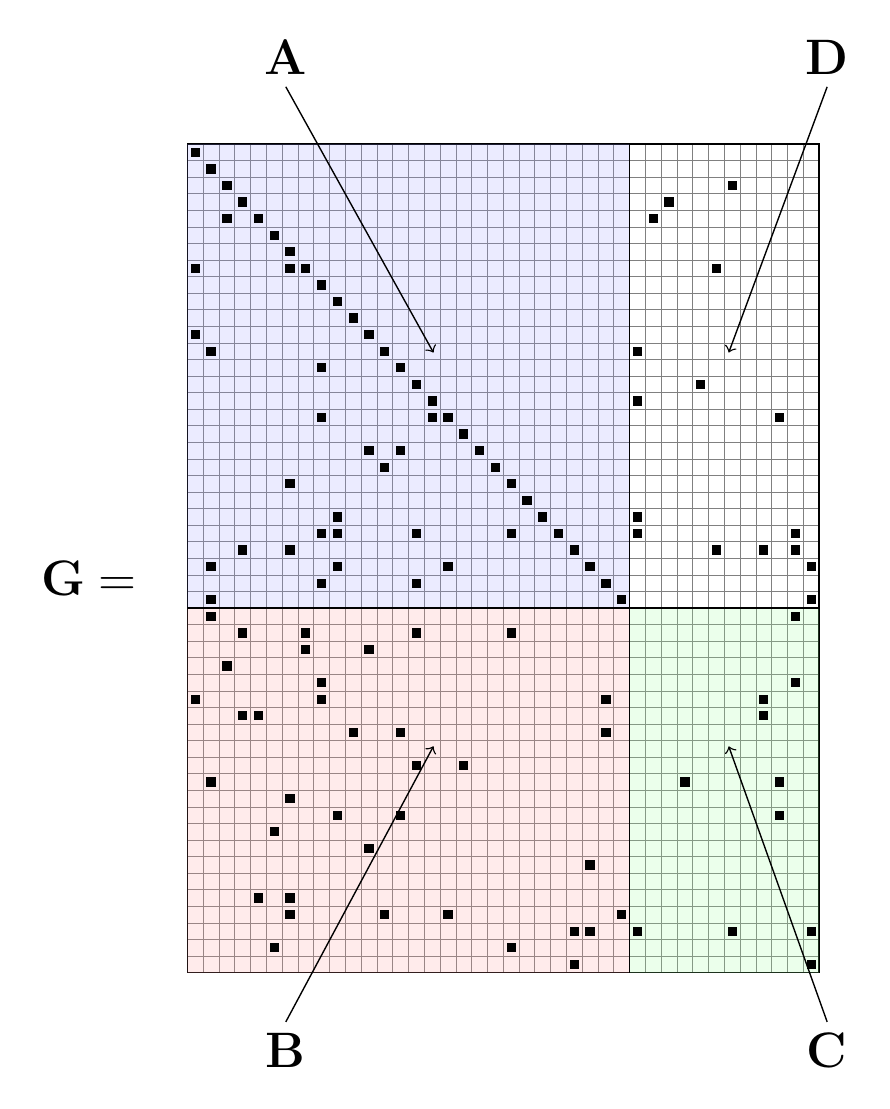}
                \label{fig:piv_a}
                }\hspace{2cm}
      \subfigure[Structure of $\mathbf{G}'$ after the zero matrix procedure.]{
                \includegraphics[width=0.67\columnwidth,draft=false]{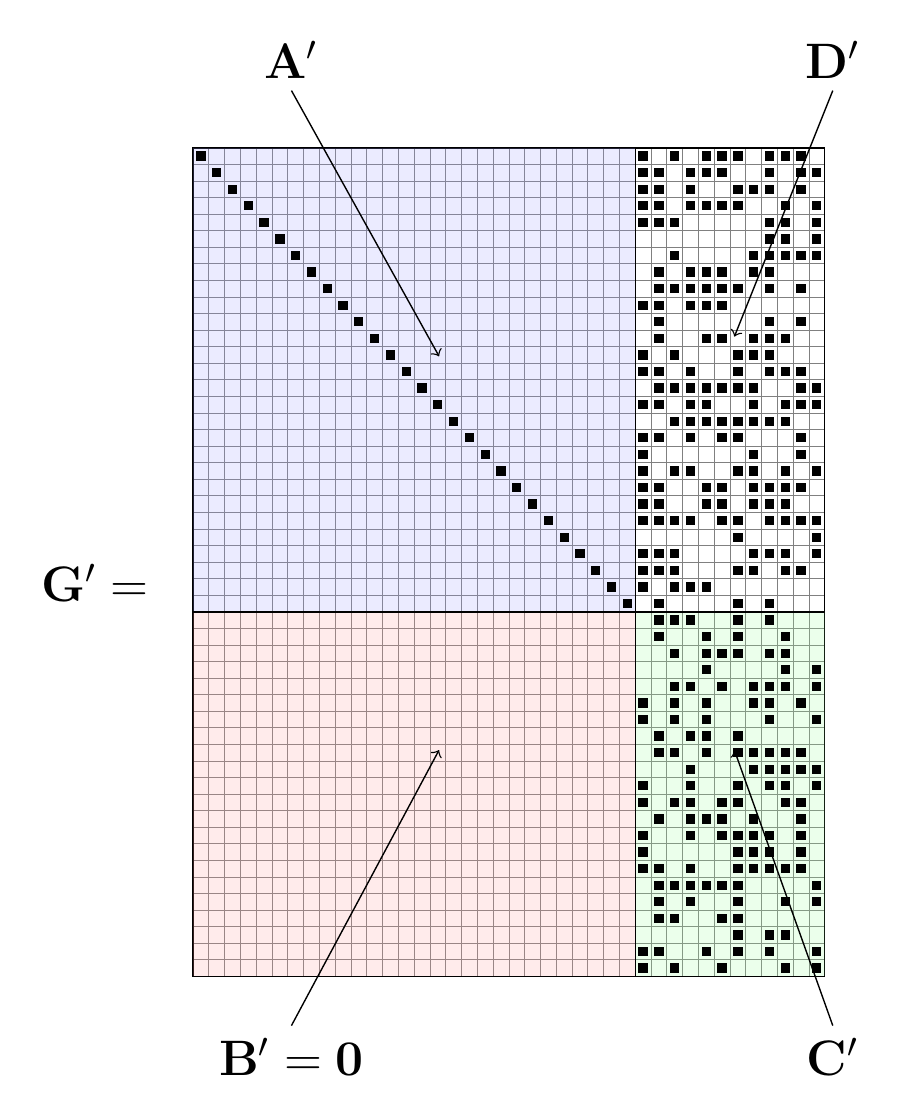}
                \label{fig:piv_b}}
    \caption{Triangulization and zero matrix procedure steps of inactivation decoding.}
\end{figure*}
\begin{enumerate}
  \item {\emph{Triangularization.} $\mathbf{G}^T$ is put in an approximate lower triangular form. At the end of this process we are left with lower triangular matrix $\Amatrix$ and matrices $\Bmatrix$, $\Cmatrix$, and $\Dmatrix$ which are sparse as shown in Fig.~\ref{fig:piv_a}. The columns of $\mathbf{G}$ corresponding to matrices $\Cmatrix$ and $\Dmatrix$ are usually referred to as inactive columns. This process consists of column and row permutations.}
  \item {\emph{Zero matrix procedure.} The matrix $\Amatrix$ is put in a diagonal form and matrix $\Bmatrix$ is zeroed out through row sums. As a consequence matrices $\Cmatrix$ and $\Dmatrix$ tend to become dense. Fig.~\ref{fig:piv_b} shows The structure of $\mathbf{G}'$ at the end of this procedure.}
  \item \emph{\ac{GE}}. \ac{GE} is applied to solve the systems of equations $\tilde{\mathbf{\osymb}} =\tilde{\mathbf{v}} \left[\Cmatrix'\right]^T$, where the symbols in $\tilde{\mathbf{v}}$ are called \emph{inactive variables} (associated with the columns of the matrix $\Cmatrix'$ in Fig.~\ref{fig:piv_b}) and  $\tilde{\mathbf{\osymb}}$ are known terms associated with the rows of the matrix $\Cmatrix'$ in Fig.~\ref{fig:piv_b}. This step drives the cost of inactivation decoding since its complexity is cubic in the number of inactivations.
  \item {\emph{Back-substitution.} Once the values of the inactive variables have been determined, back-substitution is applied to compute the values of the remaining variables in $\mathbf{v}$.}
\end{enumerate}

Decoding succeeds only if the rank of the matrix $\Cmatrix'$ equals the number of inactive variables.

In this paper we focus on the triangularization step since it is the one that determines the number of inactivations. In the remainder of the paper we will use a bipartite graph representation of the \ac{LT} code. At the left hand side of the bipartite graph we will show the source symbols and at the right hand side output symbols. The adjacency matrix of the bipartite graph is given by matrix $\mathbf{G}$, where source symbols correspond to columns and output symbols to rows of $\mathbf{G}$. Due to the one-to-one correspondence between nodes and symbols, we will use both names interchangeably.

The triangularization step can be represented by an iterative pruning of the  bipartite graph of the \ac{LT} code. At each step, a reduced graph is obtained as the sub-graph of the original LT code graph involving only a subset of the input symbols (that we call \emph{active} input symbols) and their neighbors. We will use the term \emph{reduced} degree of a node (symbol) to refer to the  degree of a node (symbol) in the reduced graph. Hence, the reduced degree of a node (symbol) is less or equal to its (original) degree. We will use the notation $\deg(c)= d$ to refer to the (original) degree of an output symbol.
Let us now introduce some additional definitions that will be used to model the triangularization step.

\begin{mydef}[Ripple] We define the ripple as the set of output symbols of reduced degree 1 and we denote it by $\rippleset$.
\end{mydef}
\noindent The cardinality of the ripple will be denoted by $\r$ and the corresponding random variable as $\Ripp$. 
\begin{mydef}[Cloud] We define the cloud as the set of output symbols of reduced degree $d\geq 2$ and we denote it by $\cloudset$.
\end{mydef}
\noindent The cardinality of the cloud will be denoted by $\c$ and the corresponding random variable as $\Cloud$.

\begin{figure}[t!]
\begin{center}
\includegraphics[width=0.75\columnwidth]{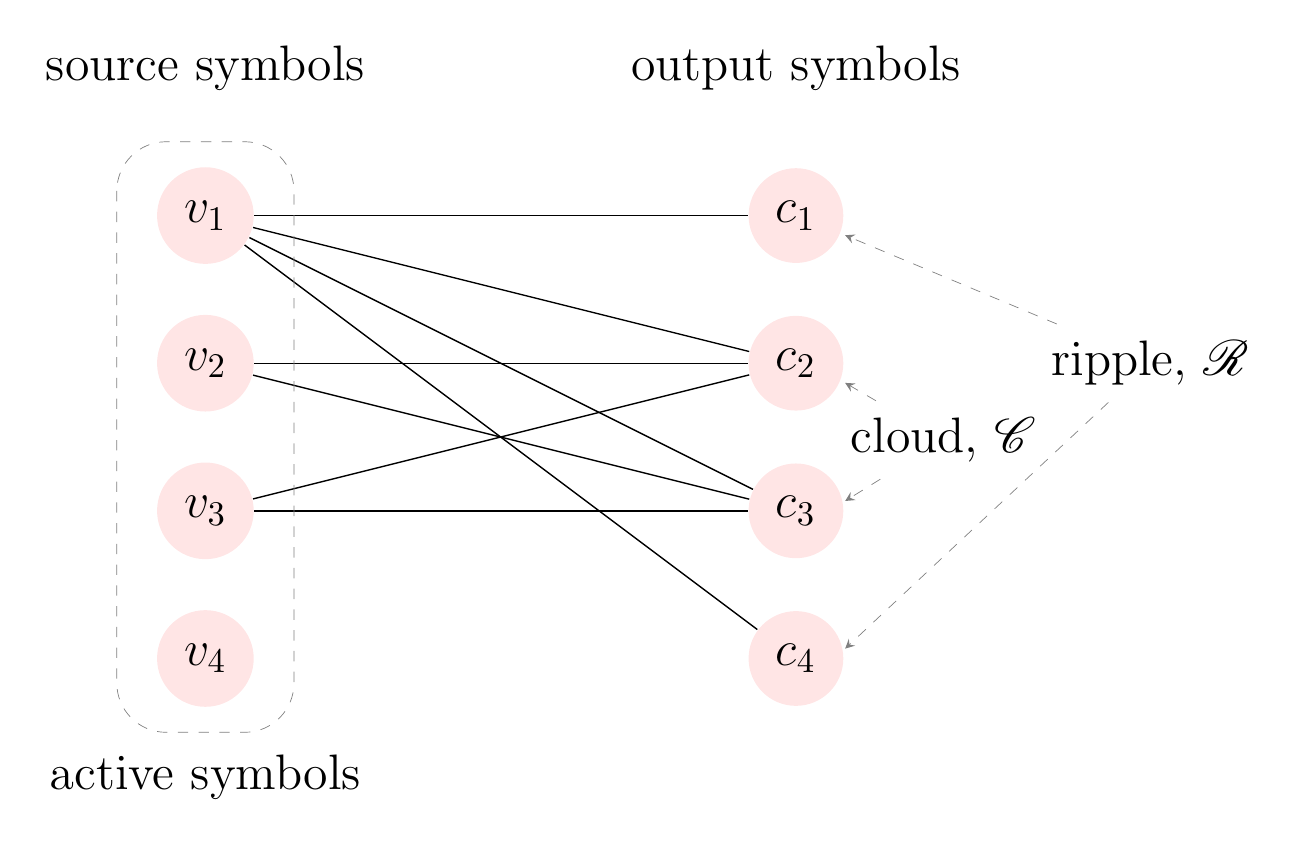}
\centering \caption{Example of bipartite graph for an LT code. Source symbols are represented at the left hand side, output symbols at the right hand side.}
\label{fig:lt_0}
\end{center}
\end{figure}

Figure~\ref{fig:lt_0} shows an example of bipartite graph with for an LT code with $4$ source symbols and $4$ output symbols. In the graph we can see how all $4$ source symbols are active. If we now look at the output symbols we can see how the ripple and the cloud are composed of two elements each,  $\rippleset=\{ c_1, c_4\}$ and $\cloudset=\{c_2, c_3\}$.
Before triangularization starts all source symbols are marked as active. At every step of the process, triangularization marks one active source symbol as \emph{resolvable} or \emph{inactive} and the symbol leaves the reduced graph. Therefore, the reduced graph will initially correspond to the bipartite graph of the \ac{LT} code. After $k$ steps the reduced graph will correspond to an empty graph.
In the following, for all the definitions provided above we will add a temporal dimension through the subscript $u$ that corresponds  to the number of active input symbols in the graph. As we will see next, since at each step of the triangularization procedure the number of active input symbols decreases by $1$, $u$ decreases as the algorithm  progresses. More specifically, the triangularization will start with $u=k$ active symbols and it will end after $k$ steps with $u=0$.
The following algorithm describes the triangularization procedure at step $u$ (i.e.,  in the transition to from $u$  to $u-1$ active symbols):

\begin{algo}[Triangularization with random inactivations]
{~}
\begin{itemize}
 \item {If the ripple $ \ripple{u}$ is not empty $(\ru>0)$}
    \begin{itemize}

        \item[] {The decoder selects an output symbol $\osymb  \in \ripple{u}$ uniformly at random.
        The only  neighbor of $\osymb$, i.e. the input symbol $v$, is marked as resolvable and leaves the reduced graph. The edges attached to $v$ are removed.}
    \end{itemize}
  \item {If the ripple $\ripple{u}$ is empty $(\ru=0)$}
    \begin{itemize}
        \item[] An inactivation takes place. One the input active symbols, $v$, is chosen uniformly at random.\footnote{This is certainly neither the only possible inactivation strategy nor the one leading to the least number or inactivations. However, this strategy makes the analysis trackable. For an overview of the different inactivation strategies we refer the reader to \cite{paolini2012}.} This input symbol is marked as inactive and leaves the reduced graph.  The edges attached to $v$ are removed.
    \end{itemize}
\end{itemize}
\label{alg:triang}
\end{algo}

\begin{figure}[h!]
        \centering
        \subfigure[LT code graph example, $u=4$ \vspace{8mm}]{
               \includegraphics[height=0.36\columnwidth,draft=false]{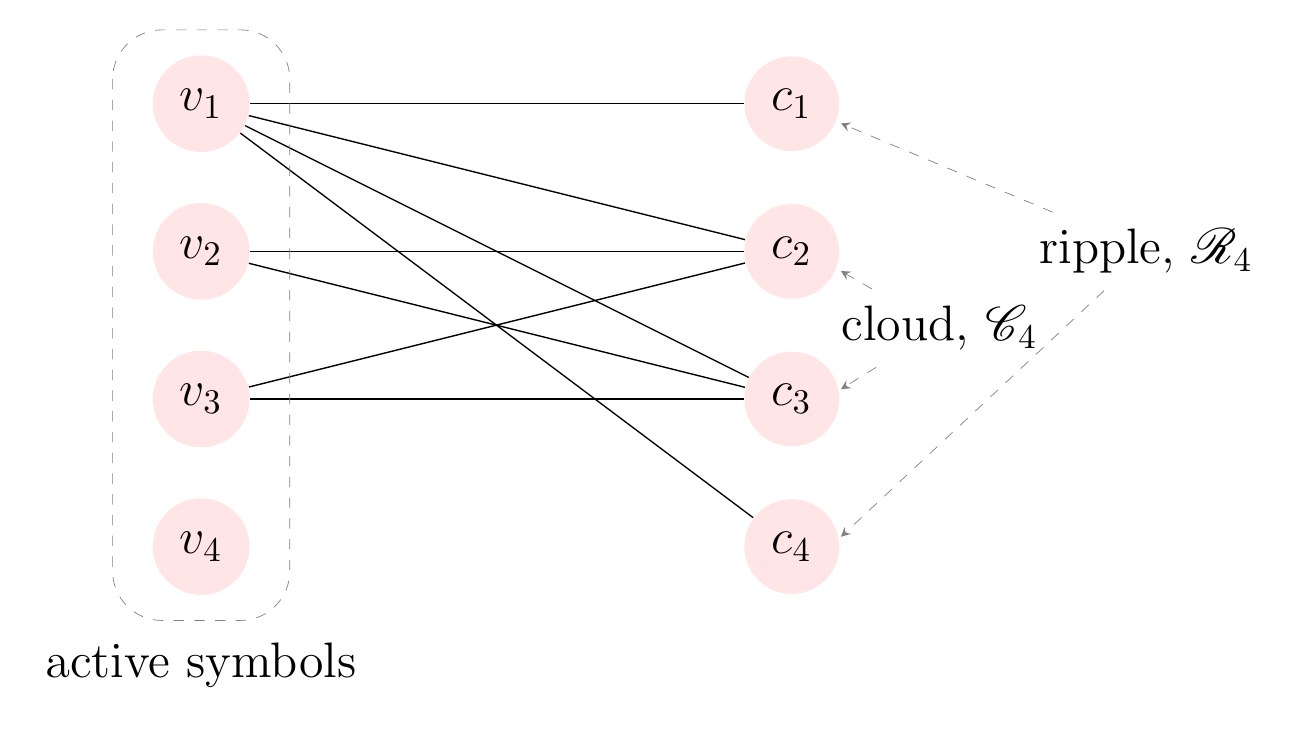}
               \label{fig:example_4}}
        \vspace{8mm}
        \subfigure[LT code graph example, $u=3$]{
               \hspace*{0.21 cm}\includegraphics[height=0.342\columnwidth,draft=false]{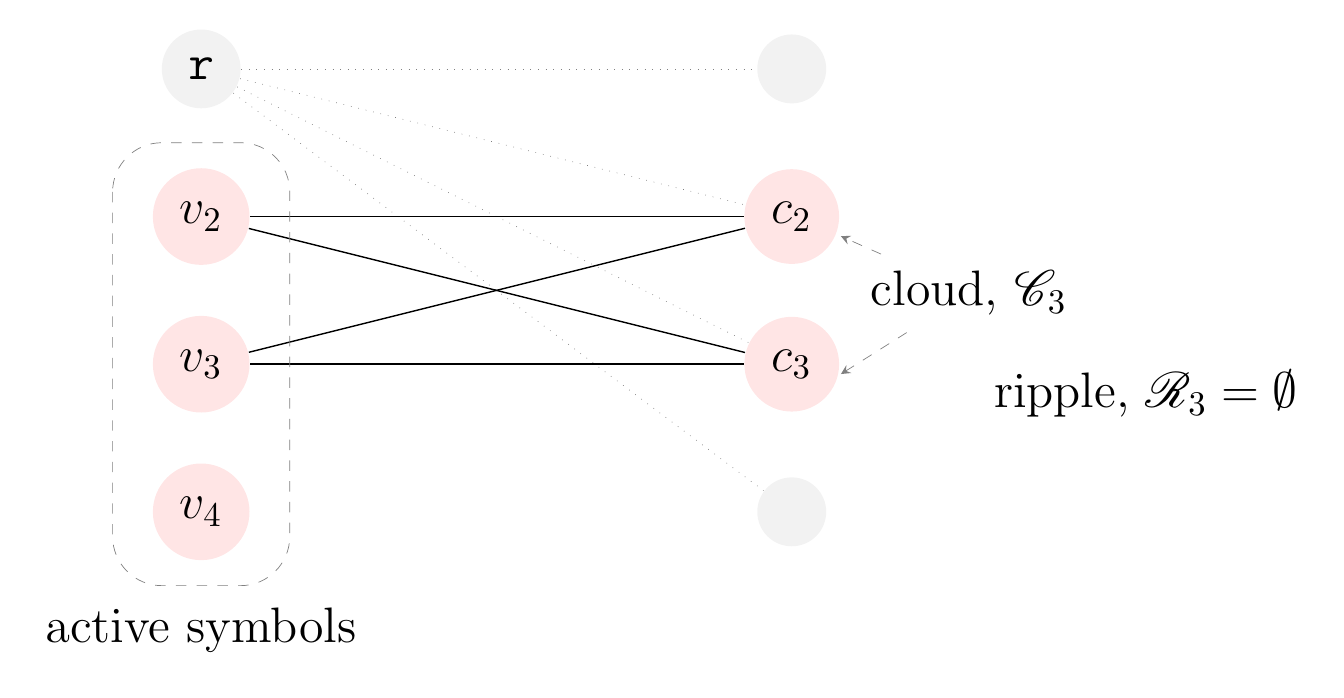}
                \label{fig:example_3}}
        \vspace{8mm}
        \subfigure[LT code graph example, $u=2$]{
               \includegraphics[height=0.35\columnwidth,draft=false]{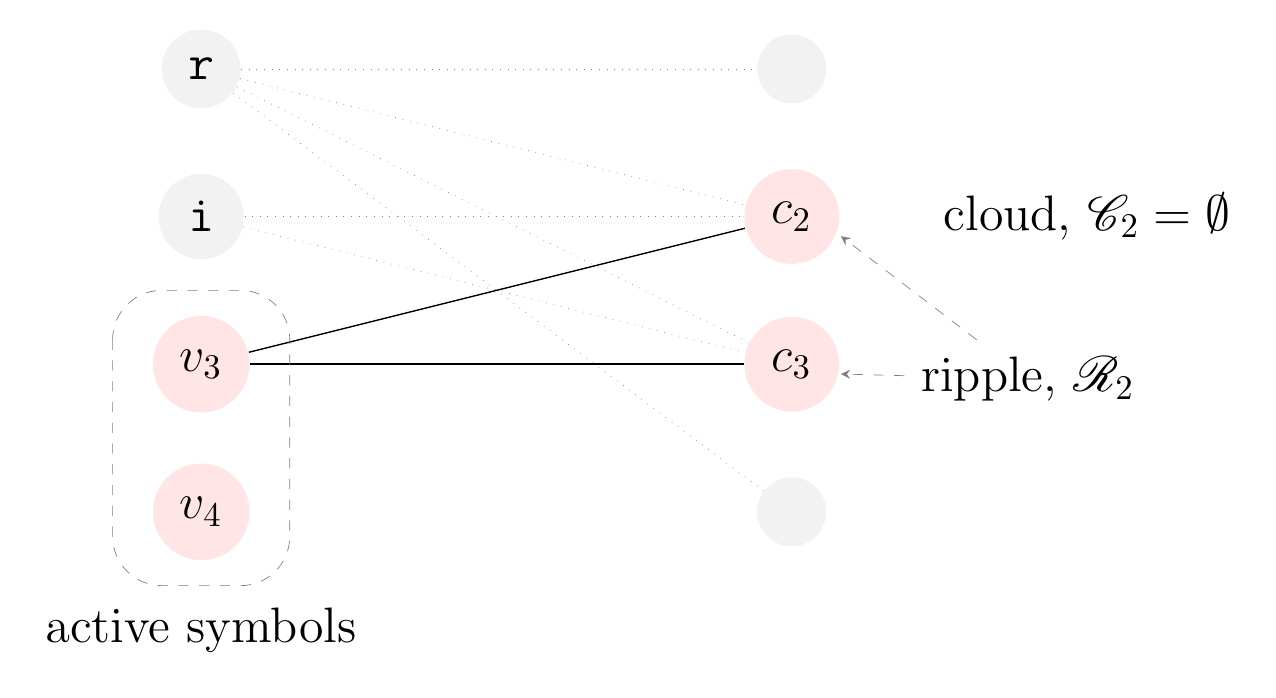}
             \label{fig:example_2}}
        \vspace{8mm}
        \subfigure[LT code graph example, $u=1$]{
               \includegraphics[height=0.35\columnwidth,draft=false]{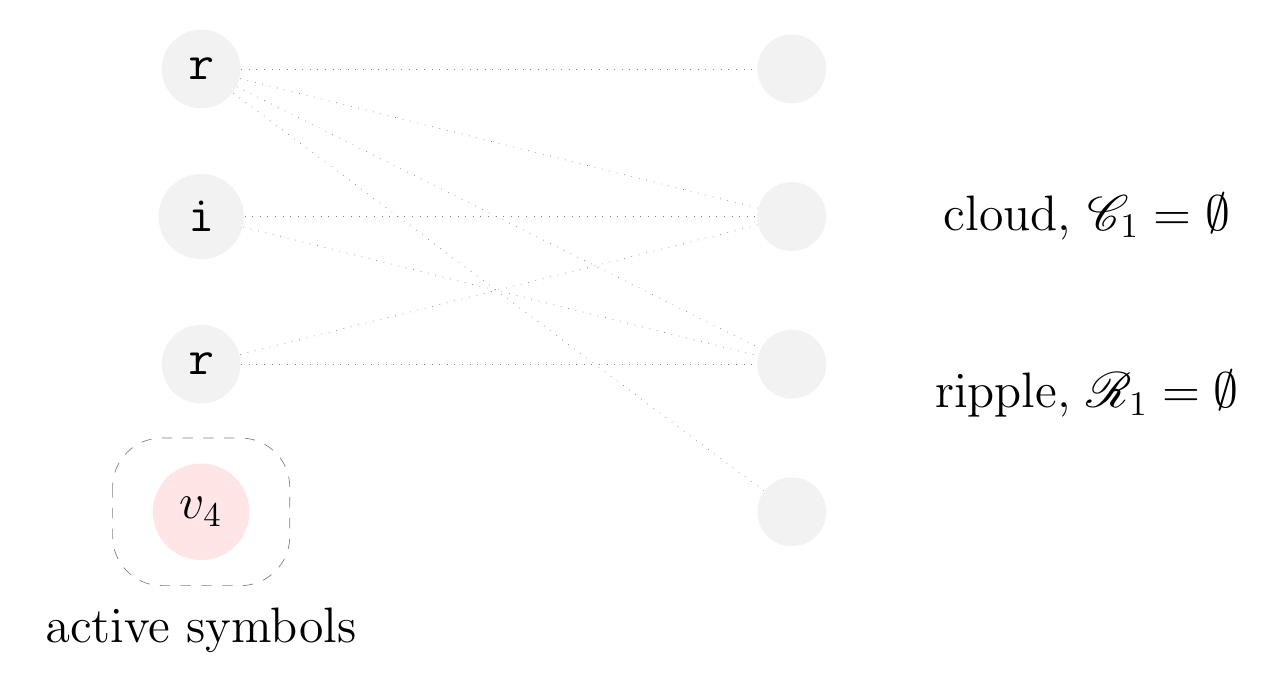}
               \label{fig:example_1}}
        \vspace{8mm}
        \subfigure[LT code graph example, $u=0$]{
               \hspace*{0.23 cm}\includegraphics[height=0.35\columnwidth,draft=false]{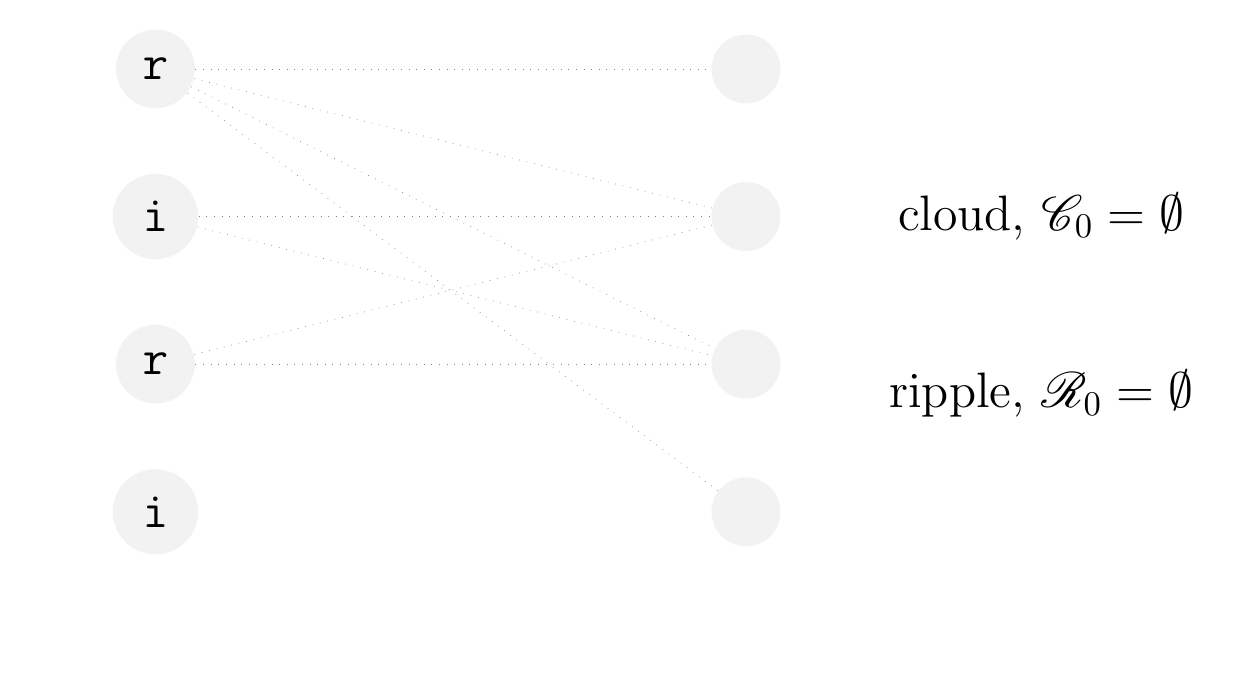}
               \label{fig:example_0}}
        \caption{Example of triangularization process for an LT code.}
        \label{fig:example}
\end{figure}

At the end of the procedure, the source symbols which are marked as resolvable  correspond to the columns of matrices $\Amatrix$ and $\Bmatrix$ in Figure~\ref{fig:piv_a}. Similarly, the source symbols marked as inactive correspond to the columns of matrices $\Cmatrix$ and $\Dmatrix$.

\begin{example}
In order to illustrate Algorithm~\ref{alg:triang} in Figure~\ref{fig:example} we provide an example for an LT code with $k=4$ source symbols and $m=4$  output symbols.
\begin{enumerate}
    \item {Transition from $u=4$ to $u=3$. In Figure~\ref{fig:example_4} we can observe
        how at the initial step $u=4$, the ripple is not empty, $\r_{4}=2$. Hence, in the transition to $u=3$ one of the source symbols ($v_1$) is marked as  resolvable, it leaves the graph and all its attached edges are removed, see Figure~\ref{fig:example_3}. The nodes $c_1$ and $c_4$ leave the graph since their reduced degree becomes zero. }
    \item{Transition from $u=3$ to $u=2$. In Figure~\ref{fig:example_3} we can see how the ripple is empty, $\r_{3}=0$. Therefore, in the transition to $u=2$ an inactivation takes place. Node $v_2$ is chosen at random and becomes inactive. All edges attached to $v_2$ are removed from the graph. As a consequence nodes $c_2$ and $c_3$ leave the cloud $\cloudset_3$ and enter the ripple $\rippleset_2$, as it can be seen in  Figure~\ref{fig:example_2}.}
    \item{Transition from $u=2$ to $u=1$. We can see in Figure~\ref{fig:example_2} how the ripple is not empty, in fact, $\r_{2}=2$. Source symbol $v_3$ is marked as resolvable and all its attached edges are removed. Nodes $c_2$ and $c_3$ leave the graph because their reduced degree becomes zero (see Figure~\ref{fig:example_1}).}
    \item{Transition from $u=1$ to $u=0$. In Figure~\ref{fig:example_1} we can see how the ripple and cloud are now empty. Hence, an inactivation takes place: node $v_4$ is marked as inactive and the triangularization procedure ends.}

\end{enumerate}
\end{example}

Let us remark that the definition of the triangularization procedure given in Algorithm~\ref{alg:triang} has the peculiarity that it never stops before $k$ steps regardless of the properties of the bipartite graph. For example, if some of the input symbols have no edge to any output symbols they are simply marked as inactive at some step, as it happened in our example with node $v_4$.

\section{First Order Analysis for Inactivation Decoding}\label{sec:model_first}

We follow the model introduced in \cite{Karp2004,shokrollahi2009theoryraptor,Maatouk:2012} for \ac{BP} decoding of LT codes. In our work we show how this approach can also be used to model inactivation decoding of \ac{LT} codes. The decoder is modelled as a finite state machine with state
\[
\S{u}:=(\Cu, \Ru ).
\]
 In this section we derive a recursion which allows to obtain $\Pr \{ \S{u-1}=(\c_{u-1}, \r_{u-1}) \}$ as a function of  $\Pr \{ \S{u}=(\cu, \ru )\}$.

We shall first analyze how the ripple and cloud change in the transition from $u$ to $u-1$ active source symbols. Since in the transition exactly one active source symbol is marked as either resolvable or inactive and all its attached edges are removed, the degree of some of the output symbols may be reduced. Consequently, some of the cloud symbols may enter the ripple and some of the ripple symbols may become of reduced degree zero and leave the reduced graph.
Let us first focus on the symbols that leave the cloud in the transition given that $\S{u}=(\c_u, \r_u)$.  Since in an \ac{LT} code the output symbols neighbors are selected uniformly at random, the number of cloud symbol which leave $\cloud{u}$ and enter $\ripple{u-1}$ is binomially distributed with parameters $\c_u$ and $p_u$, being $p_u$ the probability of a symbol leaving $\cloud{u}$ to enter $\ripple{u-1}$,
\begin{align}
p_u :&= \Pr \{ c \in \ripple{u-1} | c \in \cloud{u} \}\\[2mm] &= \frac { \Pr \{ c \in \ripple{u-1}\, , \, c \in \cloud{u} \} }  { \Pr \{ c \in \cloud{u} \}}.
\label{eq:pu_prob}
\end{align}

\medskip

 We are first interesting in evaluating
$  \Pr \{ c \in \ripple{u-1}\, , \, c \in \cloud{u} | \deg(c)= d\} $
   which corresponds to the probability that one of the edges of a degree-$d$ output symbol $c$ connected to the symbol being marked as inactive or resolvable at the transition, one edge to one of the $u-1$ active symbols after the transition and the remaining $d-2$ edges connected to the $k-u$ not active input symbols (inactive or resolvable). In other words, the symbol must have \emph{reduced} degree $2$ \emph{before} the transition and \emph{reduced} degree $1$ \emph{after} the transition.
\begin{prop}
The probability that a symbol $c$ belongs  to the cloud at step $u$  and  enters the ripple at step $u-1$,  given its original degree $d$ is given by
\begin{align}
 \Pr \{ c \in \ripple{u-1}\, , \, & c \in \cloud{u} | \deg(c)= d \} = \nonumber \\
& \mkern-65mu \frac{d}{k} (d-1)\frac{u-1}{k-1}  \frac{\binom{k-u}{d-2}}{\binom{k-2}{d-2}}
\label{eq:z_and_l_d}
\end{align}
for $d\geq 2$, while  $\Pr \{ c \in \ripple{u-1}\, , \, c \in \cloud{u} | \deg(c)= d \} = 0$ for $d<2$.
\end{prop}
\begin{proof}
The probability of an edge connecting to the the symbol being marked as inactive or resolvable at the transition is $1/k$, and there are $d$ distinct choices for the edge connected to it. This accounts for the term $d/k$ in \eqref{eq:z_and_l_d}. Moreover, there are  $d-1$ choices for the edge going to the $u-1$ active symbols after the transition, while the probability of an edge being connected to the set of $u-1$ active symbols is $(u-1)/(k-1)$. This is reflected in the term $(d-1)(u-1)/(k-1)$ in \eqref{eq:z_and_l_d}. The probability of having exactly $d-2$ edges going to the $k-u$ not active input symbols is finally
\[
\frac{\binom{k-u}{d-2}}{\binom{k-2}{d-2}}.
\]
\end{proof}
By removing the conditioning on $d$ in \eqref{eq:z_and_l_d} we obtain
\begin{align}
\Pr \{ c \in \ripple{u-1}\, , \,&  c \in \cloud{u} \} =  \\& \mathlarger {\sum}_{d=2}^{\dmax}   \Omega_d \frac{d}{k} (d-1)  \frac{u-1}{k-1}  \frac{\binom{k-u}{d-2}}{\binom{k-2}{d-2}}.
\label{eq:z_and_l}
\end{align}
The denominator of \eqref{eq:pu_prob} is given by the probability that the randomly chosen output symbol $c$ is in the cloud when $u$ input symbols are still active. This is equivalent to the probability of not being in the ripple or having reduced degree zero (all edges are going to symbols marked as inactive or resolvable) as provided by the following Proposition.
\begin{prop}
The probability that the randomly chosen output symbol $c$ is in the cloud when $u$ input symbols are still active is
\begin{align}
\Pr & \{ c \in \cloud{u}\}=\\ &   1 -  \mathlarger{\sum}_{d=1}^{\dmax}  \Omega_d  \left[ u\frac{\binom{k-u}{d-1}}{\binom{k}{d}} + \frac{\binom{k-u}{d}}{\binom{k}{d}}\right].
\label{eq:z}
 \end{align}
\end{prop}
\begin{proof}
The probability of $c$ not being in the cloud is given by is given by the  probability of $c$ being in the ripple or having reduced degree $0$. Being the two events mutually exclusive, we can compute such probability as the sum of two probabilities, the probability of $c$ having reduced degree $1$ (i.e., of $c$ being in the ripple) and the probability of $c$ having reduced degree $0$. Let us focus on the first of the two probabilities.
Assuming the degree of $c$ being $d$, the probability that $c$ has reduced degree $1$ equals the probability of $c$ having one neighbor among the $u$ active source symbols and $d-1$ neighbors among the $k-u$ non-active ones. This is given by
                               \[
                               d \frac{u}{k} \frac{\binom{k-u}{d-1}}{\binom{k-1}{d-1}},
                               \]
                               that corresponds to the first term in \eqref{eq:z}.
The probability of $c$ having reduced degree $0$ is the probability of having all $d$ neighbors in the $k-u$ non-active symbols, leading to the term
\[
\frac{\binom{k-u}{d}}{\binom{k}{d}}
\]
in \eqref{eq:z}.
\end{proof}
The probability $p_u$ can be finally obtained through \eqref{eq:pu_prob} making use of \eqref{eq:z_and_l} and of \eqref{eq:z}.

\medskip

We are interested now in analyzing the number $\erv_u$ of output symbols leaving the ripple during the transition at step $u$. We denote by $\Erv_u$ the random variable associated with $\erv_u$. We distinguish two cases.
In the first case, the ripple is not empty. In this case no inactivation takes place. Hence, an output symbol is chosen at random and removed from the ripple. Any other output symbol in the ripple which is connected to the chosen input symbol leaves the ripple during the transition. Hence, for $\ru>0$ we have
\begin{align}
\Pr\{\Erv_u=\erv_u & |\Ru=\ru\}=\\ &\binom{\ru-1}{\erv_u-1} \left(\frac{1}{u}\right)^{\erv_u-1} \mkern-3mu \left( 1- \frac{1}{u} \right)^{\ru-\erv_u}
\end{align}

\vspace{4pt}

\noindent
with $1\leq \erv_u \leq \ru$.
In the second case, the ripple is empty ($\ru=0$). Since no output symbols can leave the ripple, we have
$
\Pr\{\Erv_u=0|\Ru=0\} = 1$.
Now we are in the position to derive the transition probability
\[
\Pr\{\S{u-1}=(\c_{u-1},\r_{u-1})|\S{u}=(\c_{u},\r_{u})\}.
\]
Introducing $\b_u:=\cu-\c_{u-1}$ and observing that $\erv_u-\bu=\ru-\r_{u-1}$ we have
\begin{align}
\Pr\{\S{u-1}&=(\cu-\bu,\ru-\erv_u+\bu) | \S{u}=(\cu,\ru)\} = \nonumber \\[2mm]
& \binom{\cu}{\bu} {p_u}^{\bu} (1-p_u)^{\cu-\bu} \binom{\ru-1}{\erv_u-1} \,\,\times \nonumber \\[2mm]
&\left(\frac{1}{u}\right)^{\erv_u-1} \left( 1- \frac{1}{u} \right)^{\ru-\erv_u}
\label{eq:prob_transition}
\end{align}

\vspace{4pt}

\noindent
for $\ru>0$, while
\begin{align}
\Pr\{\S{u-1}&=(\cu-\bu,\bu) | \S{u}=(\cu,0)\} = \nonumber \\
& \binom{\cu}{\bu} {p_u}^{\bu} (1-p_u)^{\cu-\bu}.
\label{eq:prob_transition_r_0}
\end{align}

\vspace{4pt}

\noindent
Therefore, the probability of the decoder being in state $\S{u-1} =(\c_{u-1},\r_{u-1})$ can be computed recursively via \eqref{eq:prob_transition}, \eqref{eq:prob_transition_r_0} with the initial condition
\begin{equation}
 \Pr\{\S{k} =(\c_{k},\r_{k}) \} = \binom{m}{\r_{k}} \Omega_1^{\r_{k}} \left( 1-\Omega_1\right)^{\c_{k}}
\end{equation}
for all non-negative $\c_{k},\r_{k}$ such that $\c_{k}+\r_{k} = m$
where $m$ is the number of  output symbols.
Denoting by $\Y$ the random variable modelling the cumulative number of inactivations after $k$ steps, its expected value is finally given by
\begin{equation}
\mathsf{E}\left[\Y\right]= \sum_{u=1}^{k} \sum_{ \cu }  \Pr\{\S{u} =(\cu,0) \}. \label{eq:avg_inact}
\end{equation}
Figure ~\ref{fig:mbms_k_1000} shows the expected number of inactivations for a $k=1000$ LT code with output degree distribution \cite{MBMS12:raptor,luby2007rfc}
\begin{align}
\Omega(\x):&=\sum_d \Omega_d \x^d \\
&= 0.0098\x + 0.4590\x^2+ 0.2110\x^3+0.1134\x^4+ \\
&\quad  0.1113\x^{10} + 0.0799\x^{11} + 0.0156\x^{40}.
\end{align}
The chart reports the average number of inactivations obtained through  both Monte Carlo simulation and by \eqref{eq:avg_inact}, and shows a tight match between the analysis and the simulation results.

\begin{figure}[t!]
\begin{center}
\includegraphics[width=\figw]{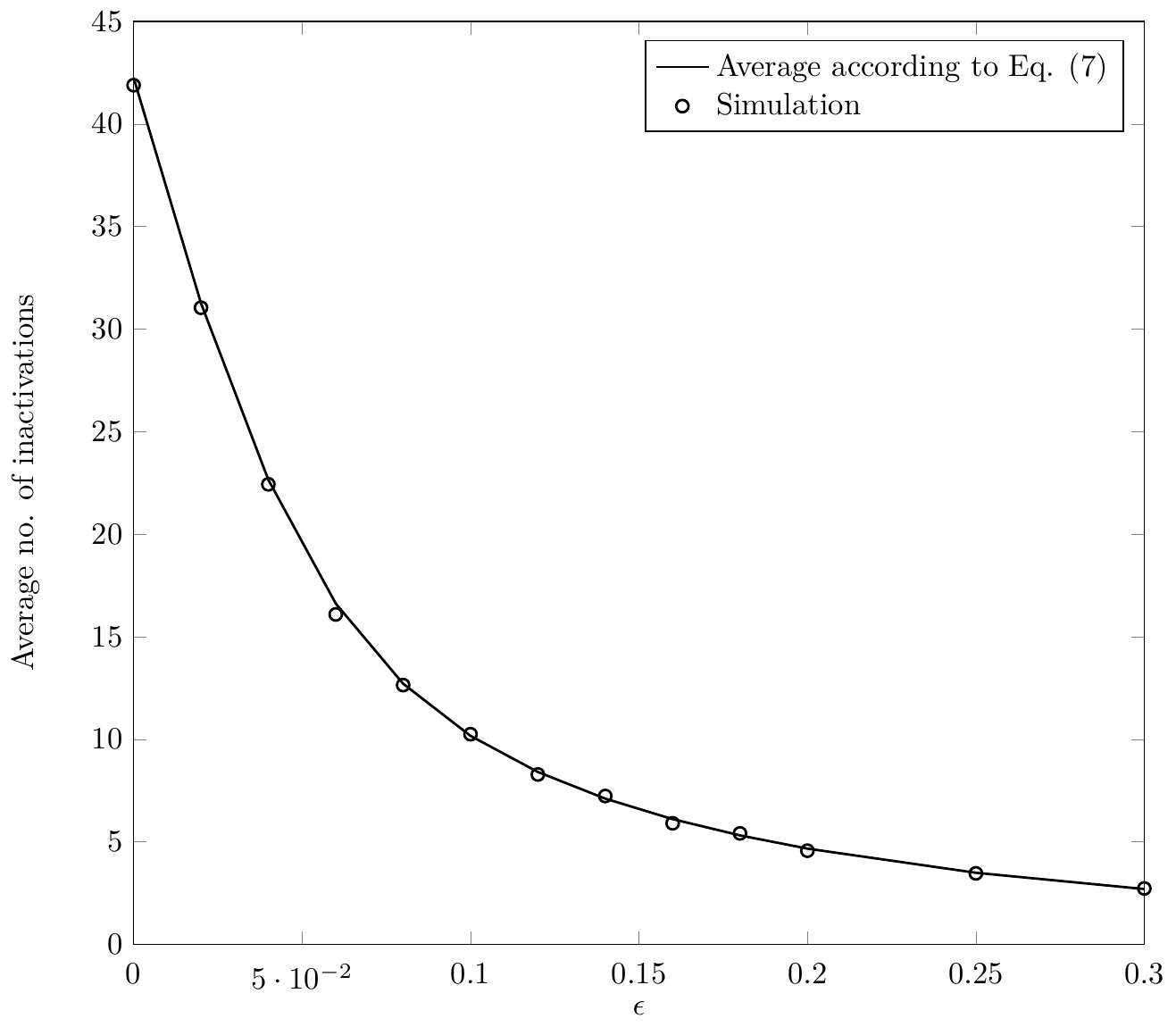}
\centering \caption{Average number of inactivations vs.  relative overhead for an \ac{LT} code with  degree distribution $\Omega(\x)
= 0.0098\x + 0.4590\x^2+ 0.2110\x^3+0.1134\x^4+ 0.1113\x^{10} + 0.0799\x^{11} + 0.0156\x^{40}$. The source block size is $k=1000$.}
\label{fig:mbms_k_1000}
\end{center}
\end{figure}

\section{Distribution of the Number of Inactivations}\label{sec:model_second}
The analysis presented in Section~\ref{sec:model_first} is able to capture the expected number of inactivations. We shall see next that the model can be easily modified to obtain the distribution of the number of inactivations.
To do so, we first need to extend the finite state machine by including in the state definition the number of inactive input symbols, i.e.,
\[
\S{u}=(\Cu, \Ru, \Nu )
\]
 with $\Nu$ being the random variable modelling the number of inactivations at step $u$. Again, we proceed by deriving a recursion which allows deriving $\Pr\{ \S{u-1}=(\c_{u-1}, \r_{u-1}, \n_{u-1} )\}$ as a function of $\Pr\{ \S{u}=(\cu, \ru, \nu )\}$.
We will look first at the transition from $u$ to $u-1$ active symbols when $\r_{u}\geq1$, that is, when no inactivation takes place. In this case the number of inactivations does not increase and we have $\n_{u-1} = \nu$. We have thus
\begin{align}
\Pr\{ \S{u-1}=(\c_{u}-\b_{u},& \r_{u}-\erv_{u}+\b_{u}, \nu) | \S{u}=(\c_u,\r_u,\nu)\} = \\[-3mm]
& \binom{\c_u}{\b_u} {p_u}^{b_u} (1-p_u)^{\c_u-\b_u}\,\,\times \\[2mm]
& \binom{\r_u-1}{\erv_u-1} \left(\frac{1}{u}\right)^{\erv_u-1} \left( 1- \frac{1}{u} \right)^{\r_u-\erv_u}.
\label{eq:prob_transition_full}
\end{align}
Let us now look at the transition from $u$ to $u-1$ active symbols when $\r{u}=0$, that is, when an inactivation takes place. In this case the number of inactivations increases by one yielding
\begin{align}
\Pr\{ \S{u-1}=(\c_{u}-& \b_{u},\b_{u}, \nu+1) | \S{u}=(\c_u,0,\nu)\} =  \\[3mm]
& \binom{\c_u}{\b_u} {p_u}^{\b_u} (1-p_u)^{\c_u-\b_u}.
\label{eq:prob_transition_r_0_full}
\end{align}

\vspace{6pt}

\noindent
The probability of the decoder being in state $\S{u-1} =(\c_{u-1},\r_{u-1}, \n_{u-1})$ can be computed recursively via \eqref{eq:prob_transition_full}, \eqref{eq:prob_transition_r_0_full} with the initial condition
\begin{equation}
 \Pr\{\S{k}=(\c_k,\r_k, \n_u) \} = \binom{m}{\r} \Omega_1^\r \left( 1-\Omega_1\right)^{\c_k} \end{equation}
for all non-negative $\c_k, \r_k$ such that $\c_k+\r_k=m$ and $\n_k=0$.

The distribution of the number of inactivations needed to complete the decoding process if finally given by
\begin{equation}
f_{\Y}(\y) = \sum_{ \c_0} \sum_{ \r_0} \Pr\{\S{0} =(\c_0,\r_0, \y) \}.\label{eq:distribution}
\end{equation}
From \eqref{eq:distribution} we may obtain the cumulative distribution $F_{\Y}(\y)$ which would give the probability of performing at most $\y$ inactivations during the decoding process. The cumulative distribution of the number of inactivations has practical implications. Assume the fountain decoder runs on a platform with limited computational capability. For example, the decoder may be able to perform a maximum number of inactivations (e.g., due to the complexity associated with the third step of the algorithm outlined in Section II, which grows cubically with the number of inactivations). Suppose the maximum number of inactivations that the decoder can handle is $\y^*$. For such a decoder, the probability of decoding failure will be lower bounded by $F_{\Y}(\y^*)$.\footnote{The probability of decoding failure is actually higher than $F_{\Y}(\y^*)$ since the system of equations to be solved in the \acf{GE} step of inactivation decoding might be rank deficient.}

Fig.~\ref{fig:mbms_dist_inact} shows the distribution of the number of inactions, for an \ac{LT} code with degree $\Omega(\x)$ from Section III and source block size $k = 300$. The distribution of inactivations has been obtained through  both Monte Carlo simulation and by \eqref{eq:distribution}, showing again a tight match between the analysis and the simulation results.

\begin{figure}[t!]
\begin{center}
\includegraphics[width=\figw]{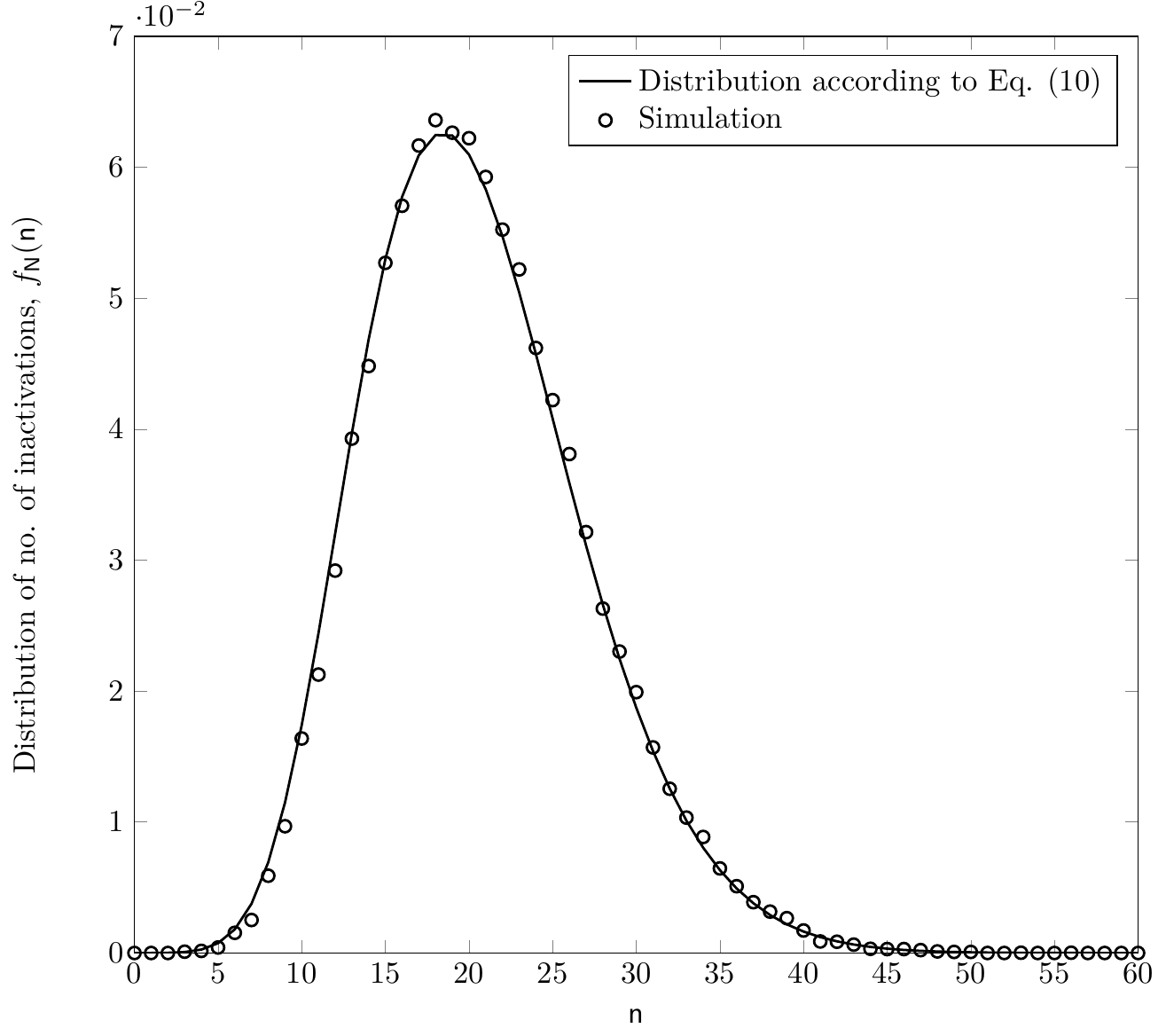}
\centering \caption{Distribution of the number of inactivations for an \ac{LT} code with  degree distribution $\Omega(\x)
= 0.0098\x + 0.4600\x^2+ 0.2110\x^3+0.1134\x^4+
  0.1110\x^{10} + 0.0800\x^{11} + 0.0156\x^{40}$. The source block size is $k=300$.}
\label{fig:mbms_dist_inact}
\end{center}
\end{figure}

\section{Conclusions and Discussion}\label{sec:Conclusions}
We have introduced a novel analysis of inactivation decoding of LT codes. A first analysis provides the expected number of inactivations needed to by an LT decoder. Furthermore, we have presented an extended analytical approach which is able to provide the distribution of the number of inactivations. The later analysis is especially important for the design of LT codes under inactivation decoding since it captures the deviations of the number of inactivations with respect to the average.

\bibliographystyle{IEEEtran}

\end{document}